\documentclass[12pt,letterpaper,english]{article}

\usepackage[margin=1in]{geometry}

\usepackage{mdwlist}

\usepackage{amsmath}
\usepackage{graphics, subfigure, float, algorithm}
\usepackage[noend]{algpseudocode}
\usepackage{fp, calc}
\usepackage{hyperref}
\usepackage[T1]{fontenc} 
\usepackage{fourier}
\usepackage{bm}

\usepackage{amscd,amsthm}

\usepackage[dvips]{graphicx,epsfig,color}
\usepackage{pst-plot}
\usepackage{pst-all}
\usepackage{pstricks-add,pst-func}
\newpsobject{showgrid}{psgrid}{subgriddiv=1,griddots=10,gridlabels=6pt}
\usepackage{comment,verbatim}

\newtheoremstyle{theorem}{1em}{1em}{\slshape}{0pt}{\bfseries}{.}{ }{}
\theoremstyle{theorem}
\newtheorem{theorem}{Theorem}
\newtheorem*{theorem*}{Theorem}

\newtheorem{lemma}[theorem]{Lemma}
\newtheorem*{lemma*}{Lemma}

\newtheorem{claim}[theorem]{Claim}

\algdef{SE}[SUBALG]{Indent}{EndIndent}{}{\algorithmicend\ }%
\algtext*{Indent}
\algtext*{EndIndent}

\theoremstyle{remark}

\newtheorem*{remark*}{Remark}

\providecommand{\setR}{\mathbb{R}}

\psset{linewidth=1pt, arrowsize=6pt}
\floatstyle{ruled}
\newfloat{algorithm}{tbp}{loa}
\floatname{algorithm}{Algorithm}

 \theoremstyle{definition}
 
  \theoremstyle{plain}
  
  \theoremstyle{plain}
  
  \theoremstyle{plain}

\theoremstyle{definition}

\newtheorem{openquestion*}{Open Question}
\theoremstyle{theorem}

\psset{nodesep=0.8pt}

\usepackage[displaymath,textmath,sections,graphics, subfigure, floats]{preview} 
\PreviewEnvironment{center}

\definecolor{augpathgray}{gray}{0.5} 

\usepackage{calrsfs} 
\DeclareMathAlphabet{\pazocal}{OMS}{zplm}{m}{n}

\makeatother
\date{}
\title{A Tale of Santa Claus, Hypergraphs and Matroids }
\author{Sami Davies\thanks{University of Washington, Seattle. Email: {\tt daviess@uw.edu}} 
\and Thomas Rothvoss\thanks{University of Washington, Seattle. Email: {\tt rothvoss@uw.edu}. 
Supported by NSF CAREER grant 1651861 and a David \& Lucile Packard Foundation Fellowship.} 
\and Yihao Zhang\thanks{University of Washington, Seattle. Email: {\tt yihaoz93@uw.edu}}}
 
\begin{document}

\maketitle

\begin{abstract}
A well-known problem in scheduling and approximation algorithms is the
Santa Claus problem. Suppose that Santa Claus has a set of gifts, and he wants to distribute them among a set of children
so that the least happy child is made as happy as possible. Here, the value that
a child $i$ has for a present $j$ is of the form $p_{ij} \in \{ 0,p_j\}$.
A polynomial time algorithm by Annamalai et al. gives a $12.33$-approximation 
and is based on a modification of Haxell's hypergraph matching argument. 

In this paper, we introduce a \emph{matroid} version of the Santa Claus problem. 
Our algorithm is also based on Haxell's augmenting tree, but with the introduction
of the matroid structure, we solve a more general problem with cleaner methods.
Our result can then be used as a blackbox to obtain a  $(6+\varepsilon)$-approximation for Santa Claus.
This factor also compares against a natural, compact LP for Santa Claus.
\end{abstract}

\section{Introduction}

Formally, the \emph{Santa Claus} problem takes as input a set $M$ of children, a set $J$ of gifts,
and values $p_{ij} \in \{ 0,p_j\}$ for all $i \in M$ and $j \in J$. In other words, a child is
only interested in a particular subset of gifts, but then its value only depends on the
gift itself. The goal is to find an assignment $\sigma : J \to M$ of gifts to children so that
$
  \min_{i \in M} \sum_{j \in \sigma^{-1}(i)} p_{ij}
$
is maximized.

The first major progress on this problem is due to Bansal and Sviridenko~\cite{SantaClaus-BansalSviridenko-STOC2006}, who showed a $O(\log \log n / \log \log \log n)$-approximation based on rounding a \emph{configuration LP}. The authors of \cite{SantaClaus-BansalSviridenko-STOC2006} also realized that in order to obtain a $O(1)$-approximation, it suffices to prove the following combinatorial statement: show that
in a uniform bipartite hypergraph with equal degrees on both sides, there is a left-perfect matching that selects a constant fraction of nodes from the original edges.
This statement was proven by Feige~\cite{ConstantIntegralityGapSantaClaus-Feige-SODA2008} using the Lov{\'a}sz Local Lemma repeatedly, but for a large unspecified constant.
Then Asadpour, Feige and Saberi~\cite{SantaClaus-AsadpourFeigeSaberi-APPROX2008} showed that
one can prove the statement from \cite{SantaClaus-BansalSviridenko-STOC2006} by using a beautiful
theorem on hypergraph matchings due to Haxell~\cite{HypergraphMatchingsHaxell95}; 
their bound\footnote{The conference version of \cite{SantaClaus-AsadpourFeigeSaberi-APPROX2008} proves a factor of 5, which was improved to 4 in the journal version~\cite{SantaClaus-AsadpourFeigeSaberi-Journal-TALG2012}.} of 4 has been slightly improved to 3.84 by Jansen and Rohwedder~\cite{ConfigurationLP-for-SantaClaus-has-gap-atmost3.84-Arxiv2018}, and then to 3.808 by Cheng and Mao~\cite{CM19}.
Recently, Jansen and Rohwedder~\cite{CompactLPforAllocationProblems-JansenRohwedder-SOSA18} also showed (still non-constructively) that it suffices to compare to a linear program with as few as $O(n^3)$ many variables and constraints, in contrast to the exponential size configuration LP.

We provide a few definitions and known results on hypergraphs matchings. 
A \emph{hypergraph} $\pazocal{H} = (X \dot{\cup} W,\pazocal{E})$ is called \emph{bipartite}
if $|e \cap X| = 1$ for all hyperedges $e \in \pazocal{E}$. A \emph{(left-) perfect matching} is a set of disjoint hyperedges $F \subseteq \pazocal{E}$ that cover each node in $X$. In general, finding perfect matchings in bipartite hypergraphs is 
$\mathbf{NP}$-hard, but there is an intriguing sufficient condition: 
\begin{theorem}[Haxell~\cite{HypergraphMatchingsHaxell95}]
Let $\pazocal{H} = (X \dot{\cup} W,\pazocal{E})$ be a bipartite hypergraph with $|e| \leq r$ for all $e \in \pazocal{E}$. Then either $\pazocal{H}$
contains a left-perfect matching, or there is a subset $C \subseteq X$ and a subset $U \subseteq W$ with $|U| \leq (2r-3) \cdot (|C|-1)$ 
so that all hyperedges incident to $C$ intersect $U$.
\end{theorem}
It is instructive to consider a standard bipartite graph, i.e., when $r=2$.
In this case, if there is no perfect matching, there is a set $C \subseteq X$ with at most 
$|C|-1$ many neighbors---thus Haxell's condition generalizes \emph{Hall's Theorem}. However, unlike Hall's Theorem, 
Haxell's proof is non-constructive and based on a possibly exponential time augmentation argument.
Only very recently (and with a lot of care), Annamalai~\cite{FindingPerfectMatchingsInHypergraphs-Annamalai-SODA16} managed to make the argument polynomial time by introducing some slack into the condition and assuming the parameter $r$ is a constant. Preceding \cite{FindingPerfectMatchingsInHypergraphs-Annamalai-SODA16}, 
Annamalai, Kalaitzis and Svensson~\cite{AlgoForSantaClaus-AnnamalaiKalaitzisSvenssonSODA15}
gave a non-trivially modified version of Haxell's argument for Santa Claus, which runs in polynomial time
and gives a $12.33$-approximation\footnote{To be precise, they obtain a $(6+2\sqrt{10}+\varepsilon)$-approximation in time $n^{O \left(\frac{1}{\varepsilon^2}\cdot \log \left (\frac{1}{\varepsilon} \right)\right )}$.}.
Our work here is concurrent with that of Cheng and Mao, who altered the algorithm of \cite{AlgoForSantaClaus-AnnamalaiKalaitzisSvenssonSODA15} to improve the approximation to 
$6 + \varepsilon$, for any constant $\varepsilon >0$~\cite{ChengM18}.
Our algorithm will also borrow a lot from \cite{AlgoForSantaClaus-AnnamalaiKalaitzisSvenssonSODA15}.
However, through a much cleaner argument, we derive a result for a more general matroid setting, while also obtaining a $6+\varepsilon$ approximation.

It should not go without mention that the version of the Santa Claus problem with arbitrary 
$p_{ij}$ has been studied under the name \emph{Max-Min Fair Allocation}. Although the integrality gap of the 
configuration LP is at least $\Omega(\sqrt{n})$~\cite{SantaClaus-BansalSviridenko-STOC2006}, 
Chakrabarty, Chuzhoy and Khanna~\cite{MaxMinFairAllocation-ChakrabartyChuzhoyKhannaFOCS09}
found a (rather complicated) $O(\log^{10}(n))$-approximation algorithm in  $n^{O(\log n)}$ time\footnote{The
approximation factor is $n^{\varepsilon}$ if only polynomial time is allowed, where $\varepsilon > 0$ is arbitrary but fixed.}.

Santa Claus has a very well studied ``dual'' minmax problem. 
Usually, it is called \emph{Makespan Scheduling} with machines $i \in M$ 
and jobs $j \in J$. Then $p_{ij}$ is the running time of job $j$ on machine $i$, and the goal 
is to assign jobs to machines so the maximum load of any machine is minimized.
In this general setting, the seminal algorithm of Lenstra, Shmoys and Tardos~\cite{SchedulingUnrelatedParallelMachines-LenstraShmoysTardos-FOCS87} gives a 2-approximation, with no
further improvement since then. A $(1.5-\varepsilon)$-approximation is $\mathbf{NP}$-hard~\cite{SchedulingUnrelatedParallelMachines-LenstraShmoysTardos-FOCS87},
and the configuration LP has an integrality gap of 2~\cite{LowerBoundOnGapOfConfigurationLP-VerschaeWiese-ESA11}. In the restricted assignment setting
with $p_{ij} \in \{ p_j,\infty\}$, the breakthrough of Svensson~\cite{MakespanScheduling-Svensson-STOC11} provides a non-constructive $1.942$-bound on the integrality gap of the
configuration LP using a custom-tailored Haxell-type search method.
Recently, this was improved by Jansen and Rohwedder~\cite{ConfigurationLPforRestrictedAssignment-JansenRohwedderSODA17} to $1.834$. In an even more restricted variant called \emph{Graph Balancing}, each job is admissable on exactly 2 machines. In this setting,  Ebenlendr, Krc{\'{a}}l and Sgall~\cite{GraphBalancingEbenlendrKS-SODA08}
gave a 1.75-approximation based on an LP-rounding approach, which has again been improved by Jansen
and Rohwedder~\cite{GraphBalancing1.749-apx-with-local-search-Arxiv2018} to 1.749 by using a local search argument.

\subsection{Our contributions}\label{sec: contributions}

Before we state our results, we provide some relevant definitions on matroids.
Let $\pazocal{M} = (X,\pazocal{I})$ be a \emph{matroid} with \emph{groundset} $X$ and 
a family of \emph{independent sets} $\pazocal{I} \subseteq 2^X$. Recall that 
a matroid is characterized by three properties: 
\begin{enumerate*}
\item[(i)] \emph{Non-emptyness}: $\emptyset \in \pazocal{I}$; 
\item[(ii)] \emph{Monotonicity}: For $Y \in \pazocal{I}$ and $Z \subseteq Y$, we have $Z \in \pazocal{I}$; 
\item[(iii)] \emph{Exchange property}: For all $Y,Z \in \pazocal{I}$ with $|Y| < |Z|$, there is an element $z \in Z \setminus Y$ so that $Y \cup \{ z\} \in \pazocal{I}$.
\end{enumerate*}
The \emph{bases} $\pazocal{B}(\pazocal{M})$ of the matroid $\pazocal{M} $ are all inclusion-wise maximal independent sets. The cardinalities of 
all bases are identical, with size denoted as $\textrm{rank}(\pazocal{M})$.
The convex hull of all bases is called the \emph{base polytope}, that is $P_{\pazocal{B}(\pazocal{M})} := \textrm{conv}\{ \chi(S) \in \{0,1\}^X \mid S\textrm{ is basis}\}$,
where $\chi(S)$ is the \emph{characteristic vector} of $S$.

Now consider a bipartite graph $G = (X \dot{\cup} W,E)$, where  $X$ is the ground set and $W$ is
a set of \emph{resources}; each resource  $w \in W$ has
a \emph{value}  $p_w \geq 0$. 
We define a problem called \emph{Matroid Max-Min Allocation}, where 
the goal is to find a basis $S \in \pazocal{B}(\pazocal{M})$, and an 
assignment of resources to that basis, $\sigma : W \to S$ with $(\sigma(w),w) \in E$, so that $\min_{i \in S} \sum_{w \in \sigma^{-1}(i)} p_w$ is maximized. 
To the best of our knowledge, this problem has not been studied before.
If $T \geq 0$ is the target objective value, 
we can define a linear programming relaxation $Q(T)$ as the set of vectors $(x,y) \in \setR_{\geq 0}^X \times \setR_{\geq 0}^E$
satisfying the constraints
\begin{eqnarray}
 x \in P_{\pazocal{B}(\pazocal{M})}; \quad 
 \sum_{w \in N(i)} p_w \cdot y_{iw} \geq T \cdot x_i \;, \forall i \in X; \quad  \sum_{i : (i,w) \in E} y_{iw} \leq 1 \;, \forall w \in W; \quad y_{iw} \leq x_i, \;  \forall (i,w) \in E. \label{eqn: lp}
\end{eqnarray} 

The decision variable $x_i$ indicates whether element $i$ should be part of the basis, and
$y_{iw}$ indicates whether resource $w$ should be assigned to element $i$. We often abbreviate 
$N(i) = \{w \in W \mid (i,w) \in E\}$ as the neighborhood of $i$.

Our main technical result is the following theorem. 
\begin{theorem} \label{thm:MainMatroidAlgorithm}
Suppose $Q(T) \neq \emptyset$. Then for any $\varepsilon>0$ one can find
 \[
   (x,y) \in Q \left (
    \left ( \frac15 - \varepsilon \right ) \cdot T - \frac15 \cdot \max_{w \in W} p_w \right )
 \] 
with both $x$ and $y$ integral in time $n^{O_{\varepsilon}(1)}$, where $n := |X| + |W|$.
This assumes that membership in the matroid can be tested in time polynomial in $n$. 
\end{theorem}
Previously this result was not even known with non-constructive methods. 
We see that Matroid Max-Min Allocation is a useful framework by applying it to the Santa Claus problem:

\begin{theorem}
The Santa Claus problem admits a $(6+\varepsilon)$-approximation  algorithm in time
$n^{O_{\varepsilon}(1)}$. 
\end{theorem}

Fix an instance of the Santa Claus problem. Let $OPT$ denote the optimal value of this instance. 
For a suitable threshold $0<\delta<1$, call a gift $j$ \emph{small} 
if $p_j \leq \delta \cdot OPT$ and \emph{large} otherwise. Then the family of sets of children that can get 
assigned large gifts forms a \emph{matchable set matroid}. We apply 
Theorem~\ref{thm:MainMatroidAlgorithm} to the \emph{co-matroid} of the matchable set matroid,
thus obtaining a basis $\{ i \in M \mid x_i=1\}$, 
which contains the children \emph{not} receiving a large gift. 
These children can receive small gifts of total value $( \frac{1}{5}- \frac{\delta}{5} -\varepsilon) \cdot OPT$, and 
the remaining children can receive a large gift with value at least $\delta \cdot OPT$.
Setting $\delta = \frac{1}{6}$ implies the claim. 
Note the approximation factor $6+\varepsilon$ is with respect to the natural, 
compact linear program in LP (\ref{eqn: lp}), which has $O(n^2)$ many variables and constraints. 
The smallest LP that was previously known to have a constant integrality gap was the $O(n^3)$-size LP of \cite{CompactLPforAllocationProblems-JansenRohwedder-SOSA18}.

\section{An algorithm for Matroid Max-Min Allocation}

In this section we provide an algorithm (see Algorithm \ref{fig:Algorithm}) that proves 
Theorem~\ref{thm:MainMatroidAlgorithm}. 

\subsection{Intuition for the algorithm}

We begin with an informal overview of our algorithm, and the formal description is in Algorithm \ref{fig:Algorithm}.
Let $G = (X \cup W,E)$ be the bipartite graph defined in Section~\ref{sec: contributions}, where we recall $X$ is the ground set of a matroid $\pazocal{M} = (X,\pazocal{I})$, and $W$ is a set of resources.
If an element of the ground set $i \in X$ has an edge $(i,j) \in E$ to to every resource $j \in U \subseteq W$,
we can consider the pair $(i,U)$ to be a hyperedge. Hyperedge $(i,U)$ \emph{covers}  $i \in X$.
For $0<\nu <1$ and val$(\cdot)$ the function summing the values in a hyperedge's resources,
we say $(i,U)$ is a \emph{$\nu$-edge} if it's a hyperedge with
minimal (inclusion-wise) resources with val$(U) = \sum_{w \in U} p_w \geq \nu \cdot  T$.
We let $\pazocal{E}_{\nu T}$ denote the set of $\nu$-edges.

Fix constants $0< \beta<\alpha<1$ and $0<\delta < 1$, to be chosen later.
The goal of the algorithm is to find a basis $S \in \pazocal{B}(\pazocal{M})$
and a hypergraph matching $M \subseteq \pazocal{E}_{\beta T}$ covering $S$. 
The algorithm is initialized with $S = \{i_0\}$, for any node $i_0 \in X$,
and $M = \emptyset$. 
We perform rank($\pazocal{M}$) many phases, 
where in each phase we find a larger matching,
where the set it covers in $X$ is independent with respect to the matroid.
In an intermediate phase, we begin with
$S \in \pazocal{I}$ and $M \subseteq \pazocal{E}_{\beta T}$
a hypergraph matching covering $S \setminus \{ i_0\}$, with one exposed node $i_0 \in X$. 
At the end of a phase, the algorithm produces an updated matching covering an independent set $S'$, with $|S'|=|S|$.
For $|S'| < \textrm{rank}(\pazocal{M})$, there exists $i_0' \in X \setminus S'$ such that $S' \cup \{i_0'\} \in \pazocal{I}$.
 Repeating this $\text{rank}(\pazocal{M})$ times, we end with a basis that is covered by $\beta$-edges.

Algorithm \ref{fig:Algorithm} generalizes the notion of an augmenting path used 
to find a maximum matchings in bipartite graphs to an \textit{augmenting tree}.
Instead of swapping every other edge in an augmenting path, as is the case for a bipartite graph, 
the algorithm swaps sets of edges in the augmenting tree to find more space in the hypergraph.
During a phase, the edges are swapped in such a way that the underlying set in $X$ covered by the matching 
is always in $\pazocal{I}$.

Edges that are candidates for being
swapped into the matching are called \textit{adding edges} (or add edges) and denoted by $A$, while those 
that are candidates for being swapped out of the matching are called \textit{blocking edges} 
and denoted by $B$. 
For hyperedges $H \subseteq \pazocal{E}_{\nu T}$ we define $H_X$ and $H_W$
as the nodes covered by $H$ in $X$ and $W$, respectively.
The parameters $\alpha$ and $\beta $ determine the value of the adding and blocking edges, respectively,
so the adding edges are a subset of $\pazocal{E}_{\alpha T}$,
while the blocking edges are a subset of $\pazocal{E}_{\beta T}$.
The algorithm introduces some slack by allowing the adding edges to contain roughly twice as many resources as the blocking edges.

Set $\delta := \max_w p_w / T$, so that
all elements in the basis receive resources with value at most $\delta \cdot T $.
The following observations follow from the minimality of the hyperedges: 
\begin{enumerate}
\item A $\nu$-edge has value less than $(\nu + \delta) \cdot T$. 
This implies an add edge has value less than $ ( \alpha + \delta ) \cdot T$, and a blocking edge has value less than $ (\beta + \delta) \cdot T$.
\item Every blocking edge has value at most $\beta \cdot T$ not covered by an add edge.
\end{enumerate}

To build the augmenting tree,
 the algorithm starts from the node in $S$ uncovered by $M$, $i_0$,
  and chooses an edge $e \in \pazocal{E}_{\alpha T}$ covering $i_0$, which is then added to $A$.
If there is a large enough hyperedge $e' \in \pazocal{E}_{\beta T}$ such that $e' \subseteq e$ and $e'$ is disjoint from $M$,
then there is enough available resources that we simply update $M$ by adding $e'$ to it.
Otherwise, $e$ does not contain a set of resources with total value $\beta \cdot T$ free from $M$. 
The edges of $M$ intersecting $e$ are added to the set of blocking edges, $B$. 
Nodes in $C=\{i_0\} \cup B_X$ are called \textit{discovered} nodes,
 as they are the nodes covered by the hypermatching $M$ that appear in the augmenting tree. 

Continuing to build the augmenting tree in later iterations, 
the algorithm uses an \textit{Expansion Lemma} to find a large set of disjoint hyperedges, 
$H \subseteq \pazocal{E}_{\alpha T}$, that cover a subset which can be swapped into $S$ in place of some subset of 
$C$, while maintaining independence in the matroid. 
The set of hyperedges $H$ either $(i)$ intersects many edges of $M$ or 
$(ii)$ has a constant fraction of edges containing a hyperedge from $\pazocal{E}_{\beta T}$ that is disjoint from $M$.

In the first case, a subset of $H$ which intersects $M$, denoted $A_{\ell+1}$, 
is added to $A$,
and the edges of $M$ intersecting $A_{\ell+1}$, denoted $B_{\ell+1}$, are added to $B$, 
for $\ell$ the index of the iteration. 
Note we naturally obtain \textit{layers} which partition the adding and blocking edges in our augmenting tree. 
The layers for the adding and blocking edges, respectively, are denoted as $A_{\ell}$ and $B_{\ell}$, and we let 
\[
A_{\leq \ell} = \bigcup\limits_{i=0}^{\ell} A_i \qquad \text{ and }  \qquad B_{\leq \ell}=\bigcup\limits_{i=0}^{\ell}B_i.
\] 
In the second case, for the set of edges $H' \subseteq \pazocal{E}_{\alpha T}$ 
that have a hyperedge from $\pazocal{E}_{\beta T}$ disjoint from $M$, 
the algorithm finds a layer with a large number of discovered nodes that can be swapped out for a subset of nodes that $H'$ covers.

At the end of each iteration, the algorithm checks whether there is any layer $\ell$ containing a large set of edges in $A_\ell$ 
with at least $\beta \cdot T$ value disjoint from $M$. If such a layer exists, the algorithm continues swapping edges into the matching from these layers.

\subsection{A detailed procedure}

See a formal description of our procedure in Algorithm~\ref{fig:Algorithm}.
Recall $\delta = \max_{w \in W} p_w/T$. 
The parameters dictating the size of the adding and blocking edges are
\[\alpha := \frac25\cdot (1 - \delta) - \frac{\varepsilon}{2} \qquad \text{ and }\qquad \beta := \frac15\cdot (1-  \delta) -\varepsilon,\]
and other parameters for the algorithm are $\kappa := \varepsilon/2$, $\phi := \frac{\alpha -\beta}{\delta + \alpha}, $
\[ \mu :=\frac{1}{1+\delta} \left (1-\alpha-\left (  \frac{\beta}{\alpha-\beta} + \kappa\right ) \cdot (\alpha+\delta) - \delta  \right), \qquad  c := \frac{(\alpha-\beta) \cdot \mu}{\beta+\delta + \kappa \cdot (\alpha-\beta)} , \qquad \gamma := \frac{1}{2\cdot \left(\frac{\log(\frac{2}{c})}{\log(1+c)}+1\right)}, \] 
for 
$0<\varepsilon \leq (1-\delta)/5$. We note that $\mu \geq  \varepsilon$, 
$\phi  \geq  \varepsilon$ and $c \geq \varepsilon/3$, for any $0  <\delta \leq 1/6$.

\begin{figure}
\hrule \vspace{2mm}
\begin{algorithmic}
  {\small 
\State {\bf Input:} Node $i_0$ and set $S \in \pazocal{I}$ with $i_0 \in S$. Matching $M \subseteq \pazocal{E}_{\beta T}$ with $M_X = S \setminus \{i_0\}$

\State Initialize: $A=A_0 = \emptyset$, $B=B_0 = \emptyset$, $C = \{i_0\}$, $\ell=0$

\smallskip
\While{TRUE} 
\State Find disjoint $H \subseteq \pazocal{E}_{\alpha T}$
covering $D \subseteq (X \setminus S) \cup C$, such that 
$|D| \geq   \mu \cdot |C| $, $(S \setminus C) \cup D \in \pazocal{I}$,  and $H_W$ is disjoint from $A_W \cup B_W$ \qquad  \textit{//\; Possible by Lemma~\ref{lem:GreedyMatchingExpansionLemma} with $W' = A_W \cup B_W$ }
\\ 
\textit{//\;Build the next layer in the augmenting tree} 
\If{$H$ intersects at least  $ \phi \cdot |H|$  many edges $M$ on $W$-side}
\State Let $B_{\ell+1} =\{e \in M : e \cap H \neq \emptyset\}$ and  $A_{\ell+1} = H$
\State Update $ B \leftarrow B \cup B_{\ell+1}$ and   $A \leftarrow A \cup A_{\ell+1}$
\State Update $C \leftarrow B_X \cup \{ i_0\}$ and  $\ell \leftarrow \ell + 1$
\\
\textit{// \; Swap sets and collapse layers}
\Else { $ H'=\{e \in H : \text{val}(e_W \setminus M_W) \geq \beta T \}$ has $  |H'| \geq \phi\cdot |H|$ } \qquad  \textit{//\; \texttt{If/else} occurs by Lemma~\ref{lem:SizeOfOverlapOfNewEdgesWithM}}
\State  For all $e \in H'$, choose one $e ' \subseteq e$ with $e' \in \pazocal{E}_{ \beta T}$ and $e'_W \cap M_W = \emptyset$;
 replace $e$ for $e'$ in $H'$
\State \textbf{Run} Algorithm \ref{fig:swap} on $M, S, H',B,$ $\ell+1$

\State \hspace{5mm}
Get outputs $\widetilde{M} \subseteq M$ covering $\widetilde{C}$, $\widetilde{H} \subseteq H'$ covering $\widetilde{D}$,  $\widetilde{\ell}$
\State Update $M \leftarrow M \setminus \widetilde{M} \cup \widetilde{H}$, $S \leftarrow S \setminus \widetilde{C} \cup \widetilde{D}$, $A \leftarrow A_{\leq \widetilde{\ell}}$, and $B \leftarrow B_{\leq \widetilde{\ell}} \setminus \widetilde{M}$
\State Update $C \leftarrow B_X \cup \{ i_0\}$ and $\ell \leftarrow \widetilde{\ell}$
\EndIf
\\
\textit{// \; Check whether other sets should be swapped and any other layers collapsed}
\Indent
\For{  all layers $i \leq \ell$}
\State  Let $A'_{i} = \{e \in A_{i} : \text{val}(e_W \setminus M_W) \geq \beta T\}$ 
\EndFor
\While { there exists a layer $\ell^* \leq \ell$ with $|A'_{\ell^*}| \geq \kappa \cdot  |B_{\ell^*}|$} \qquad  \emph{ //\;Let $\ell^*$ be the lowest such  layer} 
\State  For all $e \in A'_{\ell^*}$, choose one $e ' \subseteq e$ with $e' \in \pazocal{E}_{\beta T}$ and $e'_W \cap M_W = \emptyset$, and replace $e$ for $e'$ in $A'_{\ell^*}$
\State \textbf{Run} Algorithm \ref{fig:swap} on $M, S, A'_{\ell^*},   B,$ $\ell^*$ 
\State \hspace{5mm} Get outputs $\widetilde{M} \subseteq M$ covering $\widetilde{C}$, $\widetilde{A}_{\ell^*} \subseteq A'_{\ell^*}$ covering $\widetilde{D}$,  $\widetilde{\ell}$
\State  Update $M \leftarrow M \setminus \widetilde{M} \cup \widetilde{A}_{\ell^*}$, and $S \leftarrow S \setminus \widetilde{C} \cup \widetilde{D}$, $A \leftarrow A_{\leq \widetilde{\ell}} $ and $B \leftarrow B_{\leq \widetilde{\ell}} \setminus \widetilde{M}$
\State Update $C \leftarrow B_X \cup \{i_0\}$ and $\ell \leftarrow \widetilde{\ell}$
\EndWhile
\EndIndent
\EndWhile

\smallskip }
\end{algorithmic}
\hrule
\caption{Main algorithm} \label{fig:Algorithm}
\end{figure}

\begin{figure}
\hrule \vspace{2mm}
\begin{algorithmic}
\State {\bf Input:} 
Matching $M$, $S \in \pazocal{I}$ with $M_X = S \setminus \{i_0\}$, edges $E' \subseteq \pazocal{E}_{\beta T}$, blocking edges $B$, and layer $\ell$
\State  Let $D'$ be the nodes covered by $E'$, i.e., $D' = (E')_X$
\State Let $C' \subseteq (B_{\leq \ell-1})_X \cup \{i_0\}$ be such that $|C' |=|D' |$ and $ S \setminus C' \cup D' \in \pazocal{I}$
  \If {$i_0 \in C'$}
    \State Let $i_1 \in D'$ so that $S \setminus \{ i_0\} \cup \{ i_1\} \in \pazocal{I}$ and let $e_1 \in E'$ be the edge covering $i_1$
  \State Return $M \cup \{ e_1\}$ covering $S \setminus \{ i_0\} \cup \{ i_1\} $ and \textbf{terminate}.
\EndIf 
\State Fix layer $\widetilde{\ell} \leq \ell-1$ containing $\widetilde{C} \subseteq C' \cap (B_{\widetilde{\ell}})_X $, with $|\widetilde{C}| \geq \gamma \cdot |C'|$ \qquad \emph{//\;
 By Lemma~\ref{lem:ConstantSizeSwap}, $\widetilde{C}$ exists}
\State Let $\widetilde{D} \subseteq D'$ be such that $|\widetilde{C}| = |\widetilde{D}|$ and $S \setminus \widetilde{C} \cup \widetilde{D} \in \pazocal{I}$
\State Let $\widetilde{E}\subseteq E'$ be such that $\widetilde{E}$ covers $\widetilde{D}$, and let $\widetilde{M} \subseteq M$ be such that $\widetilde{M}$ covers $\widetilde{C}$
\State \textbf{Return}  $\widetilde{M} \subseteq M$ covering $\widetilde{C}$,  $\widetilde{E} \subseteq E'$ covering $\widetilde{D}$, and $\widetilde{\ell}$
\end{algorithmic}
\hrule
\caption{Swap subroutine} \label{fig:swap}
\end{figure}

\subsection{Correctness of the algorithm}
Here, we prove several lemmas about the performance of Algorithm \ref{fig:Algorithm}, leading to the proof of  Theorem~\ref{thm:MainMatroidAlgorithm}. 
See Figures \ref{fig:CaseIofTheAlgorithm} and \ref{fig:CaseIIofTheAlgorithm}, which illustrate the \texttt{if/else} statement of Algorithm \ref{fig:Algorithm}.

We begin by building up to our \textit{Expansion Lemma}, Lemma~\ref{lem:GreedyMatchingExpansionLemma}. 
Our algorithm takes a fixed independent set, $S$, and swaps $C \subseteq S$ out of $S$ for a set of nodes $D$,
in order to construct a new independent set of the same size.
This is possible by Lemma~\ref{lem:GreedyMatchingExpansionLemma}. 
Recall a variant of the so-called \emph{Exchange Lemma}.
 For independent sets $Y,Z \in \pazocal{I}$, let $H_{\pazocal{M}}(Y,Z)$ 
 denote the bipartite graph on parts $Y$ and $Z$ (if $Y \cap Z \neq \emptyset$, 
 then have one copy of the intersection on the left and one on the right).
For  $i \in Y \setminus Z$ and $j \in Z \setminus Y$ we insert an edge $(i,j)$ 
in $H_{\pazocal{M}}(Y,Z)$ if $Y \setminus \{ i\} \cup \{ j\} \in \pazocal{I}$. 
Otherwise, for $i \in Y \cap Z$, there is an edge between the left and right copies of $i$,
 and this is the only edge for both copies of $i$. 

\begin{lemma}[Exchange Lemma]\label{lem:ExchangeLemma}
For any matroid $\pazocal{M} = (X,\pazocal{I})$ and independent set $Y,Z \in \pazocal{I}$ with $|Y| \leq |Z|$,
the exchange graph $H_{\pazocal{M}}(Y,Z)$ contains a left perfect matching.
\end{lemma}

Next, we prove several lemmas about vectors in the base polytope with respect to sets containing swappable elements.
Lemma~\ref{lem:GreedyMatchingExpansionLemma} relies on a \textit{Swapping Lemma}, 
Lemma~\ref{lem:SumOverBasePolytopeGeneral}, for which the next lemma serves as a helper function.
\begin{lemma}[Weak Swapping Lemma] \label{lem:SumOverBasePolytope}
Let $\pazocal{M} = (X,\pazocal{I})$ be a matroid with an independent set $S \in \pazocal{I}$.
For $C \subseteq S$, define 
\[
    U := \{ i \in (X \setminus S) \cup C \mid (S \setminus C) \cup \{i\} \in \pazocal{I}\}.
\]
Then for any vector $x \in P_{\pazocal{B}(\pazocal{M})}$ in the base polytope one has $\sum_{i \in U} x_i \geq |C|$.
\end{lemma}
\begin{proof}
Note that in particular $C \subseteq U$. Moreover, an equivalent definition of $U$ is 
\[
    U = \{ i \in (X \setminus S) \cup C \mid \exists j \in C: (S \setminus \{j\}) \cup \{ i\} \in \pazocal{I}\}.
\]
Due to the integrality of the base polytope, there is a basis $B \in \pazocal{I}$ 
with $\sum_{i \in U} x_i \geq \sum_{i \in U} (\chi(B))_i = |U \cap B|$, 
where $\chi(B) \in \{ 0,1\}^X$ is the characteristic vector of $B$. 
As $S$ and $B$ are independent sets with $|S| \leq |B|$, 
from Lemma~\ref{lem:ExchangeLemma} there is a left-perfect matching in the exchange graph $H_{\pazocal{M}} (S,B)$. 
The neighborhood of $C$ in $H_{\pazocal{M}} (S,B)$ is  $U \cap B$.
As there is a left-perfect matching, $|B \cap U|$ is least $|C|$ and hence $\sum_{i \in U} x_i \geq |U \cap B| \geq |C|$.
\end{proof}

Next, we derive a more general form of the Swapping Lemma (which coincides
with the previous Lemma~\ref{lem:SumOverBasePolytope} if $D = \emptyset$): 
\begin{lemma}[Strong Swapping Lemma] \label{lem:SumOverBasePolytopeGeneral}
Let $\pazocal{M} = (X,\pazocal{I})$ be a matroid with an independent set $S \in \pazocal{I}$.
Let $C \subseteq S$ and $D \subseteq (X \setminus S) \cup C$ with $|D| \leq |C|$ and $S \setminus C \cup D \in \pazocal{I}$. Define
\[
    U := \{ i \in ((X \setminus S) \cup C) \setminus D \mid S \setminus C \cup D \cup \{i\} \in \pazocal{I}\}.
\]
Then for any vector $x \in P_{\pazocal{B}(\pazocal{M})}$ in the base polytope one has $\sum_{i \in U} x_i \geq |C| - |D|$.
\end{lemma}
\begin{proof}
Partition $C = C_1 \dot{\cup} C_2$ so that $C \cap D \subseteq C_1$, $|C_1| = |D|$ and
$S' := S \setminus C_1 \cup D \in \pazocal{I}$. 
Then note that
\begin{eqnarray*}
    U &=& \Big\{ i \in X \setminus \underbrace{(S \setminus C \cup D)}_{=S' \setminus C_2} 
    \mid \underbrace{S \setminus C \cup D}_{=S' \setminus C_2} \cup \{ i\} \in \pazocal{I} \Big\} \\
&=& \{ i \in (X \setminus S') \cup C_2 \mid S' \setminus C_2 \cup \{ i\} \in \pazocal{I}\}.
\end{eqnarray*}
Then applying Lemma~\ref{lem:SumOverBasePolytope} gives
$\sum_{i \in U} x_i \geq |C_2| = |C| - |D|.$
\end{proof}

 \smallskip
 
We bound the value of the resources in $A_W \cup B_W$ in the following claim.

\begin{claim}\label{claim: size-W'}
    For $W'= A_W \cup B_W$, at the beginning of an iteration of Algorithm \ref{fig:Algorithm}, 
    
    \[\text{val}(W' )\leq \left ( \left ( \frac{\beta}{\alpha-\beta} + \kappa\right ) \cdot (\alpha+\delta) + \delta \right ) \cdot |C|T.\]
\end{claim}
\begin{proof}
    Let $\ell$ be the highest layer in the augmenting tree so far.
    Recall the set of adding edges in layer $1<i \leq \ell$ that have value at least $\beta \cdot T$ free from $M$ is denoted
    \[A_i' = \{e \in A_i : \text{val}(e_W \setminus M_W) \geq \beta \cdot T\}.\]
    By definition of $A_i'$, every edge in $A_i \setminus A_i'$ has more than $(\alpha-\beta)\cdot T$ value in $B_i$. By the minimality of the blocking edges, each edge in $B_i$ contains value at most  $(\beta+\delta)\cdot T$.
    Therefore we see that
 \begin{eqnarray*}
        (\alpha-\beta)\cdot T \cdot  |A_{i} \setminus A_i'| &\leq& (\beta+\delta) \cdot T \cdot |B_{i}|\\
      |A_{i} \setminus A_i'| &\leq& \frac{\beta+\delta}{\alpha-\beta} \cdot |B_{i}|.
\end{eqnarray*}

From the last \texttt{while} loop in the algorithm, we know that $|A_i'| <  \kappa \cdot |B_i|$ for every layer $i$. In total we bound the value in $W_i' = (A_i)_W \cup (B_i)_W$: 
 \begin{eqnarray*}
        \text{val}(W_i') &=& \text{val}(B_i)+\text{val}(A_i' \setminus B_i)+\text{val}(A_i \setminus (A_i' \cup B_i))\\
        &=& |B_i| \cdot (\beta +\delta ) \cdot T+|A_i'| \cdot (\alpha+\delta) \cdot T+|A_i \setminus A_i'|   \cdot \beta \cdot T\\
        &\leq & |B_i| \cdot (\beta +\delta ) \cdot T+\kappa \cdot |B_i| \cdot (\alpha+\delta) \cdot T+\frac{\beta+\delta}{\alpha-\beta} \cdot |B_{i}|  \cdot \beta \cdot T\\
         &= &  \left ( \beta +\delta +\kappa \cdot (\alpha+\delta)+\frac{\beta+\delta}{\alpha-\beta} \cdot \beta  \right ) \cdot |B_i|T\\
         &= &  \left ( \left ( \frac{\beta}{\alpha-\beta} + \kappa\right ) \cdot (\alpha+\delta) + \delta \right ) \cdot |B_i|T.
 \end{eqnarray*}
Summing this up over all $i$, and using that $\sum_{i=1}^\ell |B_i| \leq |C|$ and $W' = A_W \cup B_W = \cup_{i=1}^\ell W_i'$ gives
     \[\text{val}(W')= \sum_{i=1}^\ell \text{val}(W_i')=  \left ( \left ( \frac{\beta}{\alpha-\beta} + \kappa\right ) \cdot (\alpha+\delta) + \delta \right ) \cdot |C|T.\]
\end{proof}

    

\begin{lemma}[Expansion Lemma] \label{lem:GreedyMatchingExpansionLemma}
Let $C \subseteq S \in \pazocal{I}$, $W' \subseteq W$ with $\text{val}(W') \leq \big (\frac{\beta}{\alpha-\beta} + \kappa  \big ) \cdot (\alpha + \delta) \cdot T \cdot |C| + \delta \cdot T \cdot |C|$.
Further, let $\mu = \frac{1}{1+\delta} \cdot \left (1-\alpha-\delta-(\alpha+\delta) \cdot \left (\frac{ \beta}{\alpha-\beta}+\kappa\right)\right )>  \varepsilon > 0$,
 and assume that there exists $(x,y) \in Q(T)$. 
Then there is a set $D \subseteq (X \setminus S) \cup C$ of size $|D| \geq \lceil \mu \cdot |C| \rceil$
covered by a matching $H \subseteq \pazocal{E}_{\alpha T}$, so that $H_W \cap W' = \emptyset$ and $(S \setminus C) \cup D \in \pazocal{I}$.
\end{lemma}
\begin{proof}
Note that $D$ may contain elements from $C$.
Greedily choose $D$ and the matching $H$ with $|D| = |H|$ one node/edge after the other. 
Suppose the greedy procedure gets stuck --- 
no edge can be added without intersecting $W' \cup H_W$. 
For the sake of contradiction assume this happens when $|D| < \mu |C|$. First, let 
\[
    U := \{ i \in ((X \setminus S) \cup C) \setminus D \mid (S \setminus C) \cup D \cup \{i\} \in \pazocal{I}\}
\]
be the nodes which could be added to $D$ while preserving independence. 
Then for our fixed $x \in P_{\pazocal{B}(\pazocal{M})}$, by Lemma~\ref{lem:SumOverBasePolytopeGeneral} one has
\[
\sum_{i \in U} x_i 
\geq |C| - |D|
    > (1-\mu) \cdot |C|.
\]

Let $W'' = W' \cup H_W$, for $W' = A_W \cup B_W$, be the right hand side resources
that are being covered by the augmenting tree. 
Using the minimality of the adding and blocking edges and Claim \ref{claim: size-W'},
\begin{eqnarray*}
     \text{val}(W'')  \leq \left (\mu \cdot (\alpha + \delta) + \left (  \frac{\beta}{\alpha-\beta} + \kappa\right ) \cdot (\alpha+\delta) + \delta   \right ) \cdot |C| \cdot T.  
\end{eqnarray*}
By the assumption that the greedy procedure is stuck, there is no edge $e \in \pazocal{E}_{\alpha T}$
with $e_X \in U$ and $e \cap W'' = \emptyset$. If $N(i)$ denotes the neighborhood of $i \in X$ in the bipartite
graph $G$, then this means that val$(N(i) \setminus W'') < \alpha T$ for all $i \in U$.
For every fixed $i \in U$ we can then lower bound the $y$-weight going into $W''$ as 
\begin{eqnarray*}
     \sum_{(i,w) \in E: w \in W''}{p_w} \cdot y_{i,w} &=& \underbrace{\sum_{w \in N(i)} {p_w} \cdot y_{i,w}}_{\geq T x_i} - \sum_{(i,w) \in E: w \notin W''} {p_w} \cdot \underbrace{y_{i,w}}_{\leq x_i}\\
     \sum_{(i,w) \in E: w \in W''}{p_w} \cdot y_{i,w} &\geq& T \cdot {x_i} -{x_i} \cdot \underbrace{\Big(\sum_{(i,w) \in E: w \notin W''} {p_w} \Big)}_{<\alpha T}\geq T \cdot x_i \cdot (1-\alpha). 
\end{eqnarray*}
Double counting the $y$-weight between $U$ and $W''$ by using the bounds shows that
\[
 (1-\alpha) \cdot T \cdot \underbrace{\sum_{i \in U} x_i}_{\geq (1-\mu) |C|} {\leq}
\sum_{(i,w) \in E: i \in U,w \in W''}{p_w} \cdot y_{i,w} \leq \sum_{w \in W''} {p_w} \cdot \underbrace{\sum_{i : (i,w) \in E} y_{iw}}_{\leq 1} \leq  \text{val}(W''),
\]
which simplified gives that
\[
    (1-\alpha) \cdot (1-\mu) \cdot T \cdot |C| < \left (\mu \cdot (\alpha + \delta) + \left (  \frac{\beta}{\alpha-\beta} + \kappa\right ) \cdot (\alpha+\delta) + \delta   \right )\cdot  T \cdot |C|.
\]
Rearranging the above, $\frac{1}{1+\delta} \cdot \left (1-\alpha-\left (  \frac{\beta}{\alpha-\beta} + \kappa\right ) \cdot (\alpha+\delta) - \delta  \right)< \mu$,
contradicting our choice of $\mu$.
\end{proof}

Algorithm \ref{fig:Algorithm} relies on the fact that from the set of hyperedges $H$
guaranteed by the Expansion Lemma,
there is either some constant fraction of $H$ to swap into the matching, 
or a constant fraction of $H$ is blocked by edges in the current matching. 
In the former case, significant space is found in $W$ for $S$. 
In the latter case, enough edges of the matching are intersected to guarantee the next layer in the augmenting tree is large.
The following lemma proves at least one of these conditions occurs.

\begin{lemma} \label{lem:SizeOfOverlapOfNewEdgesWithM}
Set  $\phi := \frac{\alpha - \beta}{\delta + \alpha} > 0$.
Let $M \subseteq \pazocal{E}_{\beta T}$ and $H \subseteq \pazocal{E}_{\alpha T}$ be hypergraph matchings. Further,
let 
\[
    H' := \{  e \in H \mid \text{val}(e_W \setminus M_W) \geq \beta \cdot T\}
\]
  be the edges in $H$ that still have value $\beta \cdot T$ after overlap with $M$ is removed.
 Then either (i) $|H'| \geq \phi \cdot  |H|$ or (ii) $H$  intersects at least  $\phi \cdot |H|$ edges of $M$.
  \end{lemma}
  \begin{proof}
  Let $W' = M_W \cap H_W$ be the right hand side nodes where the hypermatchings overlap and 
  suppose for the sake of contradiction that neither of the two cases occur. Then
  double counting the value of $W'$ gives
  \[ 
  \phi \cdot (\beta +\delta)\cdot T \cdot |H| > (\beta+ \delta) T \cdot \underbrace{(\#\textrm{edges in }M\textrm{ intersecting }W')}_{\phi \cdot |H| >} 
  \geq \text{val}(W')\geq \underbrace{|H \setminus H'|}_{\geq (1-\phi) \cdot |H|} \cdot (\alpha-\beta) \cdot T.
  \]
  Rearranging and simplifying, the above implies $\phi > \frac{\alpha - \beta}{\delta + \alpha}$. Thus we contradict our choice of $\phi$.
  \end{proof}

Next, we guarantee that the number of blocking edges grows geometrically.

\begin{lemma}\label{lem: exp-growth}
    At the beginning of each iteration in Algorithm \ref{fig:Algorithm}, for all $0 \leq i < \ell$, $|B_{i+1}| \geq c \cdot |B_{\leq i}|$, for 
    $c=  \frac{(\alpha-\beta) \cdot \mu}{\beta+\delta + \kappa \cdot (\alpha-\beta)}$.
\end{lemma}
\begin{proof}
We follow the proof as in \cite{AlgoForSantaClaus-AnnamalaiKalaitzisSvenssonSODA15}.

Fix a layer $i+1$, for $0 \leq i < \ell$.
    The last phase of the algorithm ensures that $|A'_{i+1}| < \kappa \cdot |B_{i+1}|$.
    Thus as least $|A_{i+1}| - \kappa \cdot |B_{i+1}|$ edges of $A_{i+1}$ are not in $A'_{i+1}$ and have at least $\alpha \cdot T - \beta \cdot  T$ nodes of $W$ in blocking edges in $B_{i+1}$.
    Simultaneously, we can upper bound 
$|(B_{i+1} \cap A_{i+1})_W |$ by $(\beta+\delta) \cdot T \cdot |B_{i+1}|$, and combining these bounds we see
    \begin{equation}
    \label{eq: exp-size}
        (\alpha \cdot T - \beta \cdot T) \cdot (|A_{i+1}| - \kappa \cdot |B_{i+1}|) \leq  (\beta+\delta) \cdot T\cdot  |B_{i+1}|.
    \end{equation}
When layer $i+1$ is first constructed, $|A_{i+1} | \geq \mu \cdot  |B_{\leq i}|.$ Further, this condition holds after all collapse phases, as layers are either removed from the tree entirely, or add edges remain in tact while blocking edges are removed. 
Subbing this lower bound on 
$|A_{i+1} | $ into Equation (\ref{eq: exp-size}),
\begin{equation*}
        (\alpha - \beta)\cdot T \cdot ( \mu \cdot |B_{\leq i}| - \kappa \cdot |B_{i+1}|)  \leq  (\beta+\delta)  \cdot T \cdot |B_{i+1}|,
    \end{equation*}
    which rearranging is
        $
        \frac{(\alpha-\beta) \cdot \mu}{\beta+\delta + \kappa \cdot (\alpha-\beta)}  \cdot |B_{\leq i}|  \leq    |B_{i+1}|.$
\end{proof}

Our last lemma will show a constant fraction of nodes
that could be swapped out of the augmenting tree come from the same layer. 
This allows us to swap out enough nodes from the same layer to make substantial progress with each iteration. 
Here, $C'$ and $\widetilde{C}$ are labeled the same as in Algorithm \ref{fig:swap}.

\begin{lemma} \label{lem:ConstantSizeSwap}
Let sets $C'$ and $\{B_i\}_{i=0}^{\ell}$ be such that $C' \subseteq (B_{\leq \ell})_X$.
Further, suppose there exists a constant $c>0$ 
such that $|C'| \geq c \cdot |B_{\leq \ell}|$
and $|B_{i+1}| \geq c \cdot |B_{\leq i}|$ for $i=0,\ldots,\ell-1$.
Then there exists a layer $0 \leq\widetilde{\ell} \leq \ell$ and constant  $\gamma = \gamma(c) > 0$,
such that $\widetilde{C} = C' \cap (B_{\widetilde{\ell}})_X$ has size
  $|\widetilde{C}| \geq \gamma  \cdot |C'|$.
\end{lemma}
\begin{proof}
By  Lemma \ref{lem: exp-growth},  $|B_{\leq \ell}|$ can be written in terms of lower indexed sets as
\[
  |B_{\leq \ell}| \geq (1+c)^k \cdot |B_{\leq \ell-k}|,
\]
for $k=0,\ldots, \ell$, by taking $c=\frac{(\alpha-\beta) \cdot \mu}{\beta+\delta + \kappa \cdot (\alpha-\beta)}  $. Therefore, $|C'| \geq c (1+c)^k \cdot |B_{\leq \ell-k}|$. Since $c$ is a constant, take $k$ large enough so $c (1+c)^k \geq 2$, namely $k \geq \frac{\log(\frac{2}{c})}{\log(1+c)}$. Then the collection $(B_{\ell-i})_X$, for $i=0,\ldots, k$, contains at least half of $C'$, so one of them contains at least  $\gamma = \frac{1}{2(k+1)}$ of $C'$.
\end{proof}

\begin{figure}
\begin{center}
\begin{subfigure}{}
\psset{xunit=0.7cm,yunit=0.5cm}
\begin{pspicture}(-1,-2)(12,5)
\rput[c](0,3){$X$}
\rput[c](0,0){$W$}
\psellipse[linewidth=0.75pt,fillstyle=solid,fillcolor=lightgray](4.5,3)(3.5,1)\rput[c](5,3){$S$}
\psellipse[linewidth=0.75pt,fillstyle=solid,fillcolor=gray](3,3)(1.5,0.7)\rput[c](2.25,3){$C$}
\psellipse[linewidth=0.75pt,fillstyle=solid,fillcolor=lightgray](9.25,3)(1.0,0.7)\rput[c](9.25,4.25){$D$}
\psellipse[linewidth=0.75pt,fillstyle=solid,fillcolor=lightgray](3,0)(2,0.7)\rput[c](3,-1.5){$(A_{\leq \ell})_W \cup (B_{\leq \ell})_W$}
\pspolygon[linearc=0.05,linewidth=0.75pt,linestyle=dashed,opacity=0.2,fillcolor=gray,fillstyle=solid](6.5,3.75)(5.75,-0.25)(7.25,-0.25)\rput[c](6.5,1.5){$B_{\ell+1}$}
\pspolygon[linearc=0.05,linewidth=0.75pt,linestyle=dashed,opacity=0.2,fillcolor=gray,fillstyle=solid](6.5,3.7)(8.5,-0.25)(9.85,-0.25)
\pspolygon[linearc=0.05,linewidth=0.75pt,linecolor=black,opacity=0.3,fillcolor=gray,fillstyle=solid](9.2,3.6)(6.5,-0.4)(8.75,-0.4)
\pspolygon[linearc=0.05,linewidth=0.75pt,linecolor=black,opacity=0.3,fillcolor=gray,fillstyle=solid](9.35,3.6)(9.25,-0.4)(11.35,-0.4)\rput[l](10.5,1.5){$A_{\ell+1}$}
\multido{\N=1.5+0.5}{7}{\cnode*(\N,0){2pt}{A}}
\multido{\N=5.5+0.5}{12}{\cnode*(\N,0){2pt}{A}}
\cnode*(3,3){2pt}{i0}\nput[labelsep=2pt]{0}{i0}{$i_0$}
\cnode*(6.5,3){2pt}{B1}
\cnode*(7.0,3){2pt}{B1}
\cnode*(9.0,3){2pt}{A1}
\cnode*(9.5,3){2pt}{A1}
\psbrace[rot=90,ref=1C,nodesepB=7pt,braceWidthInner=3pt,braceWidthOuter=3pt](5.75,-0.6)(7.25,-0.6){$\beta T$}
\psbrace[rot=90,ref=1C,nodesepB=7pt,braceWidthInner=3pt,braceWidthOuter=3pt](9.25,-0.6)(11.25,-0.6){$\alpha T$}
\end{pspicture}
\caption{Case 1 of the algorithm, where a set $A_{\ell+1} \subseteq \pazocal{E}_{\alpha T}$ of hyperedges is found that intersects many new edges $B_{\ell+1} \subseteq (M \setminus B_{\leq \ell})$. In particular $|B_{\ell+1}| \geq \Omega_{\varepsilon}(|C|)$. Note that $D$ might contain nodes from $C$.\label{fig:CaseIofTheAlgorithm}}
\end{subfigure}
\begin{subfigure}{}
\psset{xunit=0.7cm,yunit=0.5cm}
\begin{pspicture}(-1,-1.5)(12,5.5)
\rput[c](0,3){$X$}
\rput[c](0,0){$W$}
\psellipse[linewidth=0.75pt,fillstyle=solid,fillcolor=lightgray](4.5,3)(3.5,1)\rput[c](7.0,3){$S$}
\psellipse[linewidth=0.75pt,fillstyle=solid,fillcolor=gray](4.5,3)(1.85,0.7)\rput[c](5.8,3){$C$}
\psellipse[linewidth=0.75pt,fillstyle=solid,fillcolor=lightgray](4.75,3)(0.8,0.6)\rput[c](4.75,3){$\widetilde{C}$}
\psellipse[linewidth=0.75pt,fillstyle=solid,fillcolor=lightgray](9.25,3)(0.8,0.6)\rput[c](9.25,4.25){$\widetilde{D}$}
\psellipse[linewidth=0.75pt,fillstyle=solid,fillcolor=lightgray](4.0,0)(3,0.7)\rput[c](4,-1.5){$(A_{\leq \ell})_W \cup M_W$}
\pspolygon[linearc=0.05,linewidth=0.75pt,linecolor=black,opacity=0.3,fillcolor=gray,fillstyle=solid](9.1,3.6)(7.6,-0.4)(9.20,-0.4)
\pspolygon[linearc=0.05,linewidth=0.75pt,linecolor=black,opacity=0.3,fillcolor=gray,fillstyle=solid](9.4,3.6)(9.30,-0.4)(10.90,-0.4)\rput[l](10.5,1.5){$\widetilde{H}$}
\multido{\N=1.5+0.5}{11}{\cnode*(\N,0){2pt}{A}}
\multido{\N=8.00+0.50}{6}{\cnode*(\N,0){2pt}{A}}
\cnode*(3.6,3){2pt}{i0}\nput[labelsep=1pt]{180}{i0}{$i_0$}
\cnode*(9.0,3){2pt}{A1}
\cnode*(9.5,3){2pt}{A2}
\psbrace[rot=90,ref=1C,nodesepB=7pt,braceWidthInner=3pt,braceWidthOuter=3pt](9.30,-0.6)(10.90,-0.6){$\beta T$}
\end{pspicture}
\end{subfigure}
\caption{Case 2 of the algorithm, where $\widetilde{H} \subseteq \pazocal{E}_{\beta T}$ of size $|\widetilde{H}| \geq \Omega_{\varepsilon}(|C|)$ is found so that $(i)$ $\widetilde{H}$ is disjoint on the $W$-side to the matching $M$ and the adding edges in the augmenting tree, $(ii)$ $\widetilde{H}$ covers a set $\widetilde{D}$ with $S \setminus \widetilde{C} \cup \widetilde{D} \in \pazocal{I}$, and $(iii)$ $\widetilde{C}$ is from one layer of the augmenting tree. Here $\widetilde{D}$ and $\widetilde{C}$ do not have to be disjoint.\label{fig:CaseIIofTheAlgorithm}}
\end{center}
\end{figure}
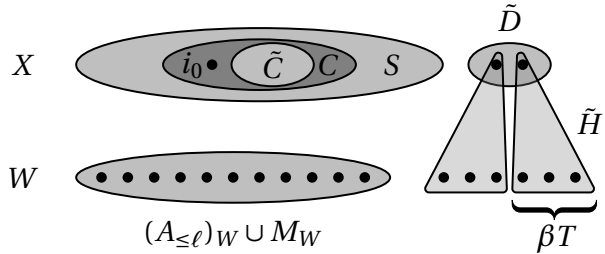

\subsection{Termination and runtime}

  As seen in Lemma~\ref{lem:ConstantSizeSwap}, 
$
|X| \geq |B_{\leq \ell}|\geq \left(1+c   \right)^{\ell} \cdot |B_0|,
$
and solving for $\ell$ shows $\frac{\log(|X|)}{\log \left(1+c   \right)}  \geq \ell$.
Thus the total number of layers at any step in the algorithm is $O(\log|X|)$. 
After each collapse of the layers, the matching $M$, and possibly the independent set $S$, are updated. 
However, the fixed exposed node $i_0$ will remain in $S$ until the very last iteration, in which the algorithm finds
an edge $e_1$ to augment the matching. 
Before we begin discussing the proof guaranteeing our algorithm terminates, we need a lemma to compare the number of blocking edges after a layer is collapsed to the number of blocking edges at the beginning of the iteration. 

\begin{lemma}\label{lem:ExtraSteps}
  Let $\widetilde{\ell}$ be the index of the collapsed layer and let $B'$ be the updated blocking edges after a collapse step. Then,
  $|B'_{\leq \widetilde{\ell}}| \leq |B_{\leq \widetilde{\ell}}| \cdot \max \left \{1- \gamma \cdot \phi \cdot \mu , 1-\gamma \cdot \kappa \cdot c \right  \}$.
\end{lemma}
\begin{proof}
 Recall $B'_{\widetilde{\ell}} = B_{\widetilde{\ell}} \setminus \widetilde{M}$ for $\widetilde{M}$ the edges of $M$ covering $\widetilde{C}$. Further, the blocking edges in layers indexed less than $\widetilde{\ell}$ are not effected in the iteration. Hence $ |B'_{\leq \widetilde{\ell}}| = |B'_{\leq \widetilde{\ell}-1}| + |B'_{ \widetilde{\ell}}|=|B_{\leq \widetilde{\ell}-1}| + |B'_{ \widetilde{\ell}}| $.
 
Then we examine a collapsed layer by itself. 

If layer $\widetilde{\ell}$ is collapsed within the \texttt{if/else} statement, we use
Lemmas \ref{lem:ConstantSizeSwap}, \ref{lem:SizeOfOverlapOfNewEdgesWithM}, and \ref{lem:GreedyMatchingExpansionLemma} to see that    
 \[|\widetilde{M}| \geq  \gamma  \cdot |C'|\geq \gamma  \cdot \phi \cdot |H|  \geq |C'|\geq \gamma  \cdot \phi \cdot \mu \cdot |C| \geq \gamma  \cdot \phi \cdot \mu \cdot
 |B_{\leq \widetilde{\ell}}|,\] 
where we recall $C'$, $H$, and $C$ are as in the \texttt{if/else} statement in Algorithm \ref{fig:Algorithm}.
Rearranging, $|B'_{ \widetilde{\ell}}| = |B_{ \widetilde{\ell}}|-|\widetilde{M}| \leq |B_{ \widetilde{\ell}}| -   \gamma \cdot \phi \cdot \mu |B_{ \leq \widetilde{\ell}}|.$
 Substituting back into $|B'_{\leq \widetilde{\ell}}|$,
\begin{eqnarray*} 
  |B'_{\leq \widetilde{\ell}}| &\leq& |B_{\leq \widetilde{\ell}-1}| + |B_{ \widetilde{\ell}}| -  \gamma \cdot \phi \cdot \mu \cdot  |B_{ \leq \widetilde{\ell}}| \\
  &=& |B_{\leq \widetilde{\ell}}| - \gamma \cdot \phi \cdot \mu \cdot |B_{ \leq \widetilde{\ell}}| =|B_{ \leq \widetilde{\ell}}| 
  \cdot \big (1-\gamma \cdot \phi \cdot \mu \big).
\end{eqnarray*}

Otherwise,  layer $\widetilde{\ell}$ is collapsed because there is some $\ell^* > \widetilde{\ell}$ where 
$|A_{\ell^*}'|$ is large compared to $|B_{\ell^*}|$ (see the last \texttt{while} loop in Algorithm \ref{fig:Algorithm}). We see that $|\widetilde{M}| \geq \gamma \cdot \kappa \cdot |B_{\ell^*}|\geq \kappa \cdot \gamma \cdot c \cdot |B_{ \leq \widetilde{\ell}}|$, where in the first inequality we use Lemma \ref{lem:ConstantSizeSwap} and in
the last we use Lemma \ref{lem: exp-growth} with the fact that $\ell^* > \widetilde{\ell}$.
So in total
\[
  |B'_{\leq \widetilde{\ell}}| \leq |B_{\leq \widetilde{\ell}}| \cdot  \left ( 1-\gamma \cdot \kappa \cdot c \right ).
\]
\end{proof}

To prove the algorithm terminates in polynomial time, we consider a signature vector $s = (s_0,s_1,\ldots, s_{\ell}, \infty)$, 
where $s_j = \lfloor \log_{b}|B_{\leq j}| \rfloor $ for $b = 1/\max \left \{1- \gamma \cdot \phi \cdot \mu ,  1-\gamma \cdot \kappa \cdot c \right \}$. 
The signature vector and proof that the algorithm terminates is inspired by \cite{AlgoForSantaClaus-AnnamalaiKalaitzisSvenssonSODA15}.

\begin{lemma}\label{lem:LexOrderDecreases}
The  vector $s$ decreases lexicographically after each iterative loop in Algorithm \ref{fig:Algorithm}. 
\end{lemma}
\begin{proof}
Let $s = (s_0,\ldots,s_{\ell},\infty)$ be a signature vector at the beginning of a step in the algorithm, and let $s'$ be the result of $s$ through one iteration of the algorithm. For $\ell+1$ denoting the newest built layer in the algorithm and $H$ the newest set of hyperedges, if $H$ intersects at least $ \phi \cdot |H|$ many edges of $M$, then another layer in the augmenting tree is built and no layer is collapsed. We have $s' = (s_0,\ldots,s_{\ell},s_{\ell+1},\infty)$ is lexicographically smaller than $s$.

Otherwise, at least one layer $0 \leq \widetilde{\ell} \leq \ell$ is collapsed. We fix $\widetilde{\ell}$ to be the last layer collapsed in each iteration of the outer \texttt{while} loop.  All finite coordinates above $s_{\widetilde{\ell}}$ are deleted from the signature vector, and all coordinates before $s_{\widetilde{\ell}}$ are unaffected. So it suffices to check that $s'_{\widetilde{\ell}} < s_{\widetilde{\ell}}$. Let $B'$ be the updated blocking edges after a collapse step. As $B_{\widetilde{\ell}}$ is the only set of blocking edges in $B_{\leq \widetilde{\ell}}$ affected by the collapse, by Lemma~\ref{lem:ExtraSteps} we have  $|B'_{\leq \widetilde{\ell}}| \leq |B_{\leq \widetilde{\ell}}| \cdot \max \left \{1- \gamma \cdot \phi \cdot \mu ,  1-\gamma \cdot \kappa \cdot c \right  \}$. Taking a $\log$ we compare the coordinates
\[
s'_{\widetilde{\ell}} = \left \lfloor \log_b \left ( \left |B'_{\leq \widetilde{\ell}} \right | \right ) \right \rfloor 
\leq 
\left \lfloor \log_b \left ( \left | B_{\leq \widetilde{\ell}} \right | \right ) \right \rfloor -1 = s_{\widetilde{\ell}}-1.
\]
\end{proof}
Choose the infinite coordinate to be some integer larger than $\log |X|$.  Since for every layer $\ell$, we have $|B_{\leq \ell}| \leq |X|$, then every coordinate of the signature vector is upper bounded by $U = O(\log |X|)$. Recall the number of layers, and thus the number of coordinates in the signature vector, is also upper bounded by $U$. Together, these imply that the sum of the coordinates of the signature vector is at most $U^2$. 

As the signature vector has non-decreasing order, each signature vector corresponds to a partition of an integer $z \leq U^2$. On the other hand, every partition of some $z \leq U^2$ has a corresponding signature vector. Thus we apply a result of Hardy and Ramanujan to find the total number of signature vectors is $\sum_{k \leq U^2} e^{O(\sqrt{k})} = |X|^{O(1)}$. Since each iteration of the algorithm can be done in polynomial time and the signature vector decreases lexicographically after each iteration, the algorithm terminates after a total time of $n^{O_{\varepsilon}(1)}$.

\section{Application to Santa Claus\label{sec:SantaClausApplication}}

In this section, we show a polynomial time $(6+\varepsilon)$-approximation algorithm
for the Santa Claus problem. Recall that 
for a given set of children $M$, and a set of presents $J$, the Santa Claus problem asks how Santa should distribute presents to children in order to maximize the minimum happiness of any child\footnote{We assume Santa to be an equitable man---not one influenced by bribery, social status, etc.}.
Here, present $j$ is only wanted by some subset of children that we denote by $A_j \subseteq M$, and present $j$ has value $p_{j}$ to child $i \in A_j$. The happiness of child $i$ is the sum of all $p_{j}$ for presents $j$ assigned to child $i$. 
We assume w.l.o.g. to know the integral objective function value $T$ of the optimum solution,
otherwise $T$ can be found by binary search.

We partition gifts into two sets: \emph{large} gifts $J_L := \{ j \in J \mid p_j > \delta_2 \cdot  T\}$
and \emph{small} gifts $J_S := \{ j \in J \mid p_j \leq \delta_1 \cdot T\}$,
for parameters $0 < \delta_1 \leq \delta_2 < 1$ such that
all gifts have values in $[0,\delta_1 \cdot T] \cup (\delta_2 \cdot T,T]$. 
Let $P(T,\delta_1,\delta_2)$ be the set of vectors $z \in \mathbb{R}_{\geq 0}^{J \times M}$ satisfying
\begin{eqnarray*}\label{eq:CompactLPforSC}
 \sum\limits_{j \in J_S: i \in A_j} p_j z_{ij} &\geq& T \cdot \Big( 1-\sum_{j \in J_L: i \in A_j} z_{ij}\Big)  \qquad \forall i \in M \\
  \sum\limits_{i \in A_j}z_{ij} &\leq& 1 \hspace{4.0cm} \forall j \in J \\
  z_{ij} &\leq& 1-\sum_{j' \in J_L: i \in A_{j'}} z_{ij'} \hspace{1.5cm} \forall j \in J_S \; \forall i \in A_j
\end{eqnarray*}

If $n = |J| + |M|$, then this LP has $O(n^2)$ many variables and $O(n^2)$
many constraints. To see that this is indeed a relaxation, take any feasible assignment $\sigma : J \to M$ with $\sum_{j \in \sigma^{-1}(i)} p_j \geq T$ for all $i \in M$. 
Now let $\sigma : J \to M \cup \{ \emptyset \}$ be a modified 
assignment where we set $\sigma(j) = \emptyset$ for gifts that we decide to drop. For each child $i \in M$ that receives at least one large gift 
we drop all small gifts and all but one large gift. 
Then a feasible solution $z \in P(T,\delta_1,\delta_2)$ is obtained by letting  
\[ 
  z_{ij} := \begin{cases} 
 1 & \textrm{if }\sigma(j) = i \\
0 & \textrm{otherwise}.
\end{cases}
\]

We will show that given a feasible solution $z \in P(T, \delta_1,\delta_2)$, there exists a feasible solution $(x^*,y^*)$ to $Q(T)$. To do this, we will exploit two underlying matroids in the Santa Claus problem, allowing us to apply Theorem~\ref{thm:MainMatroidAlgorithm}. Let 
\[
  \pazocal{I} = \left\{M_L \subseteq M |\; \exists \text{ left-perfect matching between } M_L\textrm{ and }J_L\textrm{ using edges }(i,j): i \in A_j\right\}, 
\]
be a family of independent sets. Then  $\pazocal{M} = (M,\pazocal{I})$ constitutes a \emph{matchable set matroid}.

We denote the \emph{co-matroid} of $\pazocal{M}$ by $\pazocal{M}^* = (M, \pazocal{I}^*)$. Recall that the  independent sets 
of the co-matroid are given by
\[
  \pazocal{I}^* = \left\{M_S \subseteq M | \; \exists  M_L \in \pazocal{B}(\pazocal{M}):  M_S \cap M_L = \emptyset \right\}.
\]

We can define a vector $x \in \setR^M$ with 
$ 
x_i = \sum_{j \in J_L: i \in A_j} z_{ij}
$
that lies in the matroid polytope of $\pazocal{M}$. This fact follows easily from the integrality of the fractional matching polytope in bipartite graphs. It is instructive to think of $x_i$ as the decision variable telling whether child $i \in M$ should receive a large present.

Unfortunately, $x$ does not have to lie in the base polytope --- in fact the sum $\sum_{i \in M} x_i$ might not even be integral. 
However, there always exists a vector $x'$ in the base polytope that covers
every child just as well with large presents as $x$ does. 
This observation can be stated for general matroids: 
\begin{lemma}\label{lem: RoundToBasePolytope}
Let $\pazocal{M} = (X,\pazocal{I})$ be any matroid and let $x$ be a point in its matroid polytope. Then in polynomial time one can find a point $x'$ in the base polytope so that $x' \geq x$ coordinate-wise. 
\end{lemma}
In fact the algorithm behind this claim is rather trivial: as long as $x \in P_{\pazocal{M}}$ is not in the base polytope, there is always a coordinate $i$ and a $\mu>0$ so that $x+\mu e_i \in P_{\pazocal{M}}$.

With the new vector $x' \in P_{\pazocal{B}(\pazocal{M})}$ at hand, we can redefine the $z$-assignments by letting
\[ 
  z_{ij}' = 
\begin{cases}
z_{ij} & \qquad x_i = 1\\ 
\frac{1-x_i'}{1-x_i} \cdot z_{ij} & \qquad x_i \neq 1.
\end{cases} 
\]
for $j \in J_S$; the new values $z_{ij}'$ for $j \in J_L$ can be obtained from the fractional matching
that corresponds to $x_i'$. Note that $0 \leq z_{ij}' \leq z_{ij}$ for $j \in J_S$. The reader should be convinced that
still $z' \in P(T,\delta_1,\delta_2)$, just that the corresponding vector $x'$ now lies in $P_{\pazocal{B}(\pazocal{M})}$\footnote{There is an alternative proof without the need to replace $x$ by $x'$. Add the constraint $\sum_{j \in J_L, i \in A_j} z_{ij} = \textrm{rank}(\pazocal{M})$
to $P(T,\delta_1, \delta_2)$. There is always a feasible integral solution satisfying this constraint. Then for any fractional solution $z \in P(T,\delta_1,\delta_2)$, the corresponding vector $x$ will immediately lie in the base polytope.}.

It is well known in matroid theory that  
the complementary vector $x^* := \bm{1} - x'$ lies in $P_{\pazocal{B}(\pazocal{M}^*)}$. 
Again, it is instructive to think of $x^*_i$ as the decision variable whether 
child $i$ has to be satisfied with small gifts.
Finally, the assignments $y^*$ are simply the restriction of $z'$ on the 
coordinates $(i,j) \in M \times J_S$.
The obtained pair $(x^*,y^*)$ lies in $Q(T)$, where the matroid in the definition 
of $Q(T)$ is $\pazocal{M}^*$.

As $Q(T) \neq \emptyset$, we can apply Theorem~\ref{thm:MainMatroidAlgorithm} 
which results in a subset $M_S \in \pazocal{B}(\pazocal{M}^*) $
 of the children and an assignment $\sigma : J_S \to M_S$, 
where each child in $M_S$ receives happiness at least 
$ \big(\frac15 - \frac{\delta_1}{5} - \varepsilon \big) \cdot T$ from the assignment of small gifts.
Implicitly due to the choice of the matroid $\pazocal{M}^*$, 
we know that the remaining children $M \setminus M_S = M_L$ can all receive one large gift
and this assignment can be computed in polynomial time using a matching algorithm. 
Overall, each child receives either one large present of value
at least $\delta_2 \cdot T$ or small presents of total value at least $(\frac{1}{5} - \frac{\delta_1}{5}-\varepsilon ) \cdot T$. 
Therefore each child receives value at least
\begin{equation}\label{eq: balanceSC}
  \min\Big\{ \Big(\frac{1}{5}-\frac{\delta_1}{5}-\varepsilon\Big) \cdot T, \delta_2 \cdot T \Big\} \geq \Big(\frac{1}{6}-\varepsilon\Big) \cdot T
\end{equation}
for the choice of $\delta_2 = \delta_1=\frac{1}{6}$. 
In some instances of Santa Claus, we can do better. 
Set $\delta_1$ so that $\delta_1 \cdot T$ is the largest gift value
that is at most $\frac16 \cdot T$, and set $\delta_2$ so that $\delta_2 \cdot T$ 
is the smallest gift value that is at least $\frac16 \cdot T$.
Then the algorithm guarantees that each child receives value at least 
as in the left hand side of Equation~\ref{eq: balanceSC}. 
When $\delta_1$ and $\delta_2$ are bounded away from $1/6$,
then the approximation improves. 
For example, when $\delta_2 \geq 1/5$ and $\delta_1 T$ is close to 0,
such as in the case where all gifts have value either $T$ or 1,
we approach a $(5+\varepsilon)$-approximation.

\section{Acknowledgements}
 This work originally claimed a $(4+\varepsilon)$-approximation factor. We are indebted to Stephen Arndt for finding a mistake in that analysis. The fix was actually quite easy, but it did worsen the approximation factor from $4+\varepsilon$ to $6+\varepsilon$. Our work was originally contemporary with the $(6+\varepsilon)$-approximation of Cheng and Mao \cite{ChengM18}, though they later improved those results to a $(4+\varepsilon)$-approximation \cite{CM19}.  We are also grateful to Hannaneh Akrami and Siyue Liu for further edits on a later version of this work.

We give a bit more detail on the error identified by Arndt.
An important property that we use for the runtime analysis of Algorithm \ref{fig:Algorithm} is that the number of  blocking edges in layer $\ell+1$ is large compared to the number blocking edges in the first $\ell$ layers, i.e., $|B_{\ell+1}| \geq c \cdot |B_{\leq \ell}|$ for $0<c<1$ a constant that is a function of $\varepsilon$. 
In our previous version of the algorithm, while the condition that $|B_{\ell+1}| \geq c \cdot |B_{\leq \ell}|$ was true when layer $\ell+1$ was first constructed, the invariant did not necessarily hold after some layers are collapsed. 
The way we overcame this issue here is by checking at the end of each iteration whether there is any layer containing a substantial number of add edges with sufficient available resources. Those add edges can be swapped into the matching. In order for this step to work, we were forced to make the add edges roughly twice as large as the the blocking edges, thus resulting in the worse approximation factor.  

The main contribution of this work is the addition of the matroid structure to an augmenting tree algorithm, which serves as a much cleaner framework to swap edges in and out of the tree.

\bibliographystyle{alpha}
\bibliography{santaclaus2}

\end{document}